\documentclass[a4paper,aps,pra, 11pt, onecolumn, superscriptaddress,nofootinbib, showpacs, notitlepage]{revtex4-1} 
\usepackage{times}
\usepackage{array}
\usepackage{amsmath}
\usepackage{amsfonts}
\usepackage{amssymb}
\usepackage{amsthm}
\usepackage{amsbsy}
\usepackage{bbold}
\usepackage{stmaryrd}
\usepackage{pifont}
\usepackage{latexsym}
\usepackage{mathtools}
\usepackage{bbold}
\usepackage{gensymb}
\usepackage[all]{xy}
\usepackage{graphicx}
\usepackage{amsthm}
\usepackage{csquotes}
\theoremstyle{plain}

\usepackage{epstopdf}
\usepackage{epsfig}
\usepackage[labelformat=default, labelsep=default, justification=centerlast,  font=footnotesize]{caption}
\usepackage{verbatim}
\usepackage{multirow}
\usepackage{rotating}
\usepackage{minibox}

\usepackage[titletoc,toc,title]{appendix}
\usepackage{enumerate}
\usepackage{enumitem}
\usepackage{color}
\usepackage[dvipsnames]{xcolor}
\usepackage[right=2.75cm, left=2.75cm, top=3cm, bottom=3cm]{geometry}

\usepackage{comment}
\usepackage[colorlinks]{hyperref}
\usepackage{sidecap}
\usepackage{bbm}
\usepackage{bm}

\DeclareCaptionJustification{justified}{\justifying}
\newtheorem*{theorem*}{Theorem}
\newtheorem*{definition*}{Definition}
\newtheorem*{postulate*}{Postulate}

\theoremstyle{named}

\begin{document}

\title{A no-go theorem for observer-independent facts}
\author{\v{C}aslav Brukner}
\affiliation{Vienna Center for Quantum Science and Technology (VCQ), Faculty of Physics, University of Vienna, Boltzmanngasse 5, A-1090 Vienna, Austria}
\affiliation{Institute of Quantum Optics and Quantum Information (IQOQI), Austrian Academy of Sciences, Boltzmanngasse 3, A-1090 Vienna, Austria}
\date{\today}

\begin{abstract}

In his famous thought experiment, Wigner assigns an entangled state to the composite quantum system made up of Wigner's friend and her observed system. While the two of them have different accounts of the process, each Wigner and his friend can in principle verify his/her respective state assignments by performing an appropriate measurement.  As manifested through a click in a detector or a specific position of the pointer, the outcomes of these measurements can be regarded as reflecting directly observable "facts". Reviewing \href{https://arxiv.org/abs/1507.05255}{arXiv:1507.05255}, I will derive a no-go theorem for observer-independent facts, which would be common both for Wigner and the friend. I will then analyze this result in the context of a newly derived theorem \href{https://arxiv.org/abs/1604.07422}{arXiv:1604.07422}, where Frauchiger and Renner prove that "single-world interpretations of quantum theory cannot be self-consistent". It is argued that "self-consistency" has the same implications as the assumption that observational statements of different observers can be compared in a single (and hence an observer-independent) theoretical framework. The latter, however, may not be possible, if the statements are to be understood as relational in the sense that their determinacy is relative to an observer. 

\end{abstract}

\maketitle

\subsection{Introduction}

%

One of the most debated situations concerning the quantum measurement problem is described in the thought experiment of so called ``Wigner's friend''. The experiment involves a quantum system and an observer (Wigner's friend) who performs measurements on this system in a sealed laboratory. A ``super-observer'' (Wigner) is placed outside the  laboratory. While for the friend the measurement outcome is reflected in a property of the device recording it (e.g. in the form of a click in a photo-detector or a certain position of a pointer device), Wigner can describe the process unitarily on the basis of the information that is in principle available to him. At the end of the process, the friend projects the state of the system corresponding to the observed outcome, whereas Wigner assigns a specific entangled state to the system and the friend, which he can verify performing a further experiment. When Wigner's friend observes an outcome, does the state collapse for Wigner too? And if not, how can we reconcile their different accounts of the process?

The thought experiment of Wigner's friend has great conceptual value, as it challenges different approaches to understanding quantum theory. In his original work~\cite{wigner}, Wigner designed the experiment to support his view that consciousness is necessary to complete the quantum measurement process. According to many-worlds interpretation~\cite{everett}, there are many copies of Wigner's friend in different "worlds". Each copy observes one outcome, a different one in each world. According to the Copenhagen, relational~\cite{rovelli} and Quantum Bayesian~\cite{chris} interpretations, the state is defined only relative to observer; relative to the friend, the state is projected, while relative to Wigner, it is in a superposition. Either way, supporters of any of these interpretation, will arrive at the same predictions in Wigner's verifying experiment. In contrast, objective collapse theories~\cite{grw,diosi,penrose} predict that the quantum state collapses when a superposed system reaches a certain threshold of mass, size, complexity etc., such that it becomes impossible to even prepare the entangled state of Wigner's friend and the system. Consequently, Wigner's state assignment can statistically be disproved repeating the verifying experiment. 

The descriptions of "what is happening inside the lab" as given by Wigner and Wigner's friend respectively will differ. This difference need not pose a consistency problem for quantum theory, for example, if one takes the view that the theory gives the physical description {\em relative} to observer and her/his measuring apparatus in agreement with Ref.~\cite{rovelli}. As long as the two observers do not exchange the information about their outcomes, they will remain separated form each other, each holding a different description of the systems with respect to their individual experimental arrangements. And if they do compare their predictions, they will agree. For example, should the friend communicate her result to Wigner, this would collapse the state he assigns to the friend and the system. This suggests that there should be no tension in accepting that, relative to their experimental arrangements, Wigner's friend in her measurement, as well as Wigner in his verifying measurement, each obtains a respective measurement outcome. Since these outcomes are usually manifested as clicks in detectors or definite positions of a pointer, they can be considered as directly accessible  "facts". Quite naturally, the question arises: Can the facts as observed by Wigner and by Wigner's friend be jointly considered as objective properties of the world, in which case we might call them "facts of the world"? What we mean with this question is whether there exists any theory, potentially different from quantum theory, where a {\it joint} probability may be assigned for Wigner's outcome and for that of his friend. 

Reviewing the results of Ref.~\cite{brukner}, I will derive a Bell type no-go theorem for observer-independent facts, showing that there can be no theory in which Wigner's and Wigner's friend' facts can jointly be considered as (local) objective properties. More precisely, I will show that the assumptions of "locality", "freedom of choice" and "universality of quantum theory" (the latter in the sense that there are no constrains of the system to which the theory can be applied) are incompatible with the assumption of observer-independent facts, i.e. under the assumptions one cannot define joint probabilities for Wigner's outcome and for that of Wigner's friend. This might indicate that in quantum theory, we can only define facts relative to an observation and an observer. I will then analyze the relation of these results to the theorem developed by Frauchinger and Renner~\cite{renner}, which proves that "single-world interpretations of quantum theory cannot be self-consistent". In particular, I will argue that implications of their "self-consistency" requirement are equivalent to that of a theoretical framework in which the truth values of the observational  statements by Wigner and Wigner's friend can be jointly assigned and then verified whether they are consistent or not. However, in the view of the no-go theorem, this in general need not be possible in a physical theory; the theory may operate only with facts relative to observer. 

It should be emphasized that  the no-go theorem applies to "facts" understood as "immediate experiences of observers"; it may refer to what various interpretations of quantum mechanics assume to be "real" (e.g. the wave function of the Universe, Bohmian's trajectories etc.) only to the extent to which these "realities" give rise to directly observable facts in terms of detector clicks or pointer positions.

\subsection{Deutsch's version of Wigner's friend experiment}

The standard description of the Wigner-friend thought experiment involves a quantum two-level system (system 1, e.g. a spin-1/2 particle) which can give rise to two outcomes upon measurement (e.g. two opposite directions when passing through a Stern-Gerlach apparatus). The outcomes are recorded by a measurement apparatus and eventually in the friend's memory (system 2). Now, Wigner is placed outside the isolated laboratory in which the experiment takes place and can perform a quantum measurement on the overall system (spin-1/2 particle + friend's laboratory). Take it that all experiments are carried out a sufficient number of times to collect statistics. 

For concreteness, suppose that a measurement of spin along $z$ is performed on a particle initially prepared in state $|x+\rangle_S=\frac{1}{\sqrt2}(|z+\rangle_S+|z-\rangle_S)$, where subscript $S$ refers to the spin. After  the measurement is completed, the measurement apparatus is found in one of many perceptively different macroscopic configurations, like different positions of a pointer along a scale. If the apparatus pointer is found in a specific position along the scale, the friend can say that the observable spin $z$ has the value "up" or "down". Note that for the present argument, we need not make any assumption about how the friend formally describes the spin and the apparatus or even if she uses quantum theory for it. All what is needed is the assumption that the friend perceives a definite outcome.

Wigner uses quantum theory to describe the friend's measurement. From his perspective, the measurement is described by a unitary transformation. The different possible spin states $|z+\rangle_S$ and $|z-\rangle_S$ are supposed to get entangled to the perceptively different macroscopic configurations of the apparatus and the parts of the laboratory including the friend's memory. The states of different macroscopic configurations are represented by orthogonal states  $|F_{z+}\rangle_F$ and $|F_{z-}\rangle_F$, respectively. We assume that the state of the composite system "spin + friend's laboratory" is given by: 
\begin{equation}
|\Phi\rangle_{SF}= \frac{1}{\sqrt2} \left(|z+\rangle_S|F_{z+}\rangle_F+|z-\rangle_S|F_{z-}\rangle_F \right),
\label{eq:wignerstate}
\end{equation}
where the particular phase (here "+") between the two amplitudes in Eq.~(\ref{eq:wignerstate}) is specified by the measurement interaction in control of Wigner. (Note that if Wigner would not know this phase due to lack of control of it, he would describe the "spin + friend's laboratory" in an incoherent mixture of the two possibilities.) Wigner can verify his state assignment~(\ref{eq:wignerstate}), for example, by performing a Bell state measurement in the basis: $|\Phi^{\pm}\rangle_{SF}=\frac{1}{\sqrt2} \left(|z+\rangle_S|F_{z+}\rangle_F \pm |z-\rangle_S|F_{z-}\rangle_F \right)$ and $|\Psi^{\pm}\rangle_{SF}=\frac{1}{\sqrt2} \left(|z+\rangle_S|F_{z-}\rangle_F \pm |z-\rangle_S|F_{z+}\rangle_F \right)$. 

The fact that the friend and Wigner have {\it different} accounts of the friend's measurement process is at the heart of the discussion surrounding the Wigner-friend thought experiment. Still the difference needs not give rise to any inconsistency in practicing quantum theory, since the two descriptions belong to two different observers, who remain separated in making predictions for their respective systems. The novelty of Deutsch's proposal~\cite{deutsch} lies in the possibility for Wigner to acquire direct knowledge on whether the friend has observed a definite outcome upon her measurement or not without revealing what outcome she has observed. The friend could open the laboratory in a manner which allowed communication (e.g. a specific message written on a piece of paper) to be passed outside to Wigner, keeping all other degrees of freedom fully isolated. Obviously, it is of central importance that the message does not contain any information concerning the specific observed outcome (which would destroy the coherence of state~(\ref{eq:wignerstate})), but merely an indication of the kind: "I have observed a definite outcome" or "I have {\em not} observed a definite outcome". If  the message is encoded in the state of system $M$, the overall state is:
\begin{equation}
|\Phi\rangle_{SFM}= \frac{1}{\sqrt2} \left(|z+\rangle_S|F_{z+}\rangle_F+|z-\rangle_S|F_{z-}\rangle_F \right) |\textrm{"I have observed a definite outcome"}\rangle_M,
\end{equation}
since the state of the message is factorized out from the total state. (I leave the option for the message "I have {\em not} observed a definite outcome" out, as it conflicts our experience of the situation that we refer to as measurement and it also can be used to violate the bound on quantum state discrimination~\cite{brukner}.)


If we assume universality of quantum theory in the sense that it can be applied at any scale, including the apparatus, the entire laboratory and even the observer's memory, we conclude that the message will indicate that the friend perceives a definite outcome and yet Wigner will confirm  his state assignment~(\ref{eq:wignerstate}). This should be contrasted to the "collapse models" by Ghirardi, Rimini and Weber~\cite{grw} or by Diosi~\cite{diosi} and Penrose~\cite{penrose}, which predict a breakdown of the quantum-mechanical laws at some scale. In presence of such a collapse, the prediction based on Wigner's state assignment will statistically deviate from the result obtained in the verification test.  

\begin{figure}[h]
\centering\includegraphics[width=9.5cm]{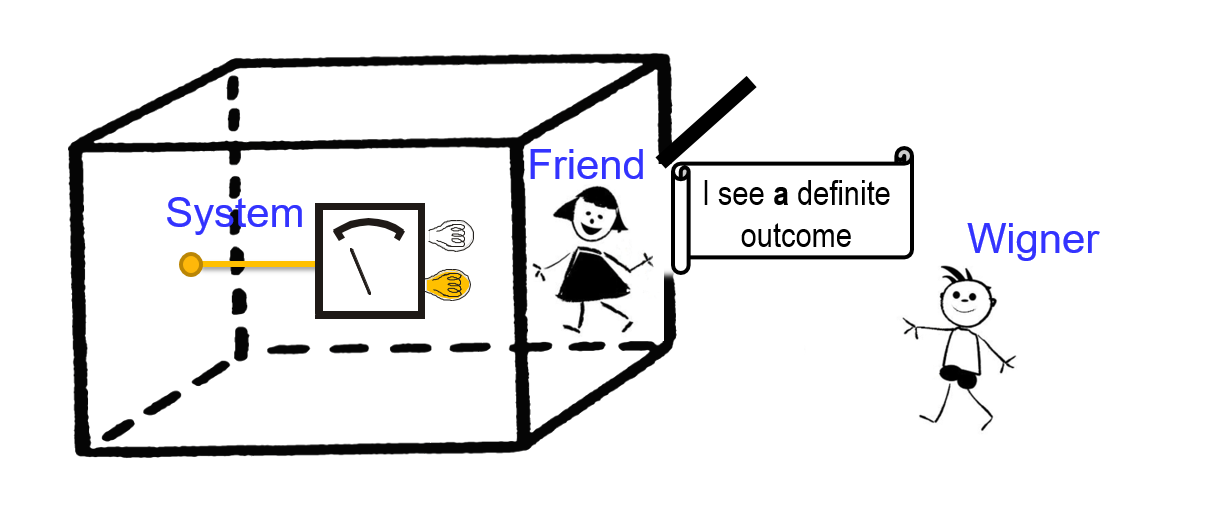}
\caption{\small{Deutsch's version of the Wigner-friend thought experiment. An observer (Wigner's friend) performs a Stern-Gerlach experiment on a spin 1/2 particle in a sealed laboratory. The outcome, either ``spin up'' or ``spin down'', is recorded in the friend's laboratory, including her memory. A super-observer (Wigner) describes the entire experiment as a unitary transformation resulting in an encompassing entangled state between the system and the friend's laboratory. The friend is allowed to communicate a message which only reports whether she sees a definite outcome or not, without in any way revealing the actual outcome she observes.}}
\label{wignersfriend}
\end{figure}


\subsection{The no-go theorem}

We have seen that Wigner not only perceives his own facts. He is also able to obtain a {\em direct evidence for the existence of the friend's facts} (although without knowing which specific outcome has been realized in the laboratory). This strongly suggests that Wigner's and Wigner's friend' facts coexist. We pose the question: Is there a theoretical framework, potentially going beyond quantum theory, in which one can account for observer-independent facts, ones that hence can be called  "facts of the world"? In such a framework, one could assign jointly truth values to both the observational statement $A_1$: "The pointer of Wigner's friend' apparatus points to result $z+$" and $A_2$: "The pointer of Wigner's apparatus points to result $\Phi$." 

One important remark: Whenever Wigner performs his measurement, he can inform the friend about the outcome he observed. Hence, Wigner's friend can learn Wigner's outcome in addition to the outcome she herself observed directly. In this way, Wigner {\it can} know the truth values of both statements $A_1$ and $A_2$. The assumption of "observer-independent facts" is a stronger condition: we require an assignment of truth values to statements $A_1$ and $A_2$ {\it independently} of which measurement Wigner performs. Wigner can either perform his verifying experiment, or he can perform Wigner's friend's measurement (for example by opening the lab, or learning it from the friend). In either experiment, the observed outcome (e.g. "$\Phi$" and "$z+$", respectively) is required to reveal the assigned truth value for $A_1$ or $A_2$. We formalize the requirement of "observer-independent facts" in the following assumption.

\begin{postulate*} {\em ("Observer-independent facts")} The truth values of the propositions $A_i$ of all observers form a Boolean algebra $\mathcal{A}$. Moreover, the algebra is equipped with a (countably additive) positive measure $p(A) \geq 0$ for all statements $A \in \mathcal{A}$, which is the probability for the statement to be true.
\label{factspostulate}
\end{postulate*}

In the proof, we will only use the conjunction of propositions of different observers, which is a weaker requirement. Furthermore, we use a countably additive measure since we are dealing with only a countable (in fact only a finite) set of elements. In Boolean algebra, one can build the conjunction, the disjunction and the negation of the statements . A typical example of a Boolean algebra is set theory. The operations are identified with the set theoretic intersection, union, and complement respectively. This is significant in the context of classical physics, where the propositions can be represented by subsets of a phase space. In the present context, one can jointly assign truth values "true" or "false" to statements $A_1$ and $A_2$ about observations made by Wigner's friend and Wigner, respectively. Moreover, one can build the conjunction $A_1 \cap A_2$ and assign  joint probability $p(A_1=\pm 1, A_2=\pm 1)$, where $A_1$ is observed by the friend and $A_2$ by Wigner (and where truth value "true" corresponds to value 1 and "false" to -1). Note that since observables corresponding to $A_1$ and $A_2$ do not commute with each other, this amounts to introducing of "hidden variables", for which we now formulate a Bell's theorem. 

\begin{theorem*}
{\em (No-go theorem for "observer-independent facts")} The followings statements are incompatible (i.e. lead to a contradiction)
\begin{enumerate}
\item \emph{"Universal validity of quantum theory"\footnote{Here we use the word "universal" in the sense of Peres~\cite{peres}: "There is nothing in quantum theory making it applicable to three atoms and inapplicable to $10^{23}$ ... Even if quantum theory is universal, it is not {\it closed}. A distinction must be made between {\it endophysical} systems -- 
those which are described by the theory -- and {\it exophysical} ones, which lie outside the domain of the theory (for example, the telescopes and photographic plates used by astronomers for verifying the laws of celestial mechanics). While quantum theory can in principle describe {\it anything}, a quantum description cannot include {\it everything}. In every physical
situation {\it something} must remain unanalyzed. This is not a flaw of quantum theory, but a logical necessity ..."}.} Quantum predictions hold at any scale, even if the measured system contains objects as large as an "observer" (including her laboratory, memory etc.).
\item \emph{``Locality''.} The choice of the measurement settings of one observer has no influence on the outcomes of the other distant observer(s).
\item \emph{``Freedom of choice''.} The choice of measurement settings is statistically independent from the rest of the experiment.
\item \emph{``Observer-independent facts''\footnote{The theorem can be derived by replacing assumptions 2, 3 and 4 with a single assumption of Bell's {\em "local causality"}. The latter already implies the existence of (local) probabilities for "joint facts" for Wigner and Wigner's friend~\cite{zb}, which is the subject of the present no-go theorem. The reason for working with the present choice of assumptions is that the relevance of the theorem for the propositions different observers make about their respective outcome become apparent.}.} One can jointly assign truth values to the propositions about observed outcomes (``facts'') of different observers (as specified in the postulate above). 
\end{enumerate}
\end{theorem*}
\begin{proof}

\begin{figure}[h]
\centering
\includegraphics[width=13.5cm]{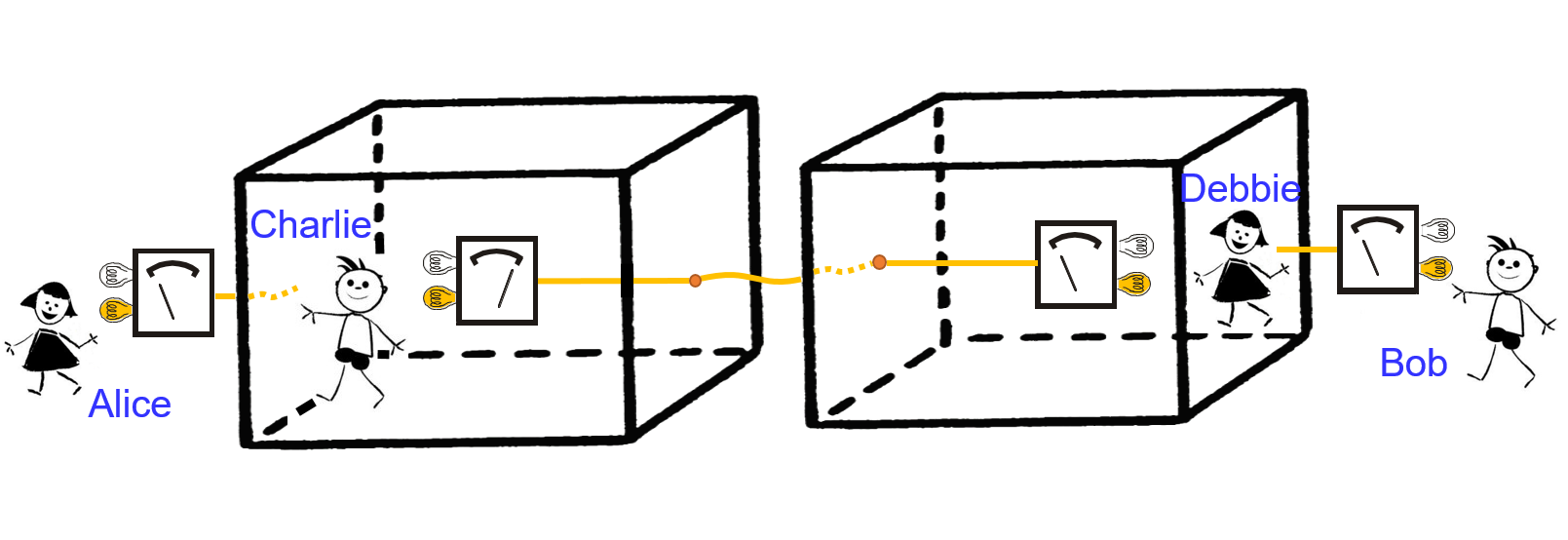}
\caption{\small {A Bell experiment on two entangled observers in a Wigner-friend scenario. The super-observers Alice and Bob perform their respective measurements on laboratories containing the observers Charlie and Debbie, who both perform a Stern-Gerlach measurement on their respective spin-1/2 particles.}}
\label{bellfig}
\end{figure}

With reference to Fig.~\ref{bellfig}, consider a pair of super-observers (Alice and Bob), who can carry out experiments on two systems that include a laboratory for each system, in each of which an observer (Charlie and Debbie, respectively) performs a measurement on a spin-1/2 particle. We consider a Bell inequality test and assume that Alice chooses between two measurement settings $A_1$ and $A_2$, and similarly Bob chooses between $B_1$ and $B_2$. The settings $A_1$ and $A_2$ correspond to the observational statements  Charlie  and Alice can make about their  respective outcomes, respectively. Similarly, the settings $B_1$ and $B_2$ correspond to observational statements of Debbie and Bob respectively. The assumptions (2), (3) and (4) together amount for the existence of local hidden variables that predefine the values for $A_1$, $A_2$, $B_1$ and $B_2$ to be $+1$ or $-1$. Moreover, the assumptions imply the existence of the joint probability $p(A_1,A_2,B_1,B_2)$ whose marginals satisfy the Clauser-Horne-Shimony-Holt inequality (CHSH): $S = \langle A_1B_1 \rangle + \langle A_1B_2 \rangle+\langle A_2B_1 \rangle-\langle A_2B_2 \rangle\leq 2$. Here, for example, $\langle A_1 B_1 \rangle = \sum_{A_1,B_1=-1,1} A_1 B_1 p(A_1, B_1)$ and $p(A_1,B_1) =\sum_{A_2,B_2=-1,1} p(A_1,A_2,B_1,B_2)$ and similarly for other cases. 

Suppose that Charlie and Debbie initially share an entangled state of two spin-1/2 particles $S_1$ and $S_2$ in a state
\begin{equation}
|\psi\rangle_{S_1S_2}= -\sin\frac{\theta}{2} |\phi^+\rangle_{S_1S_2} + \cos\frac{\theta}{2} |\psi^-\rangle_{S_1S_2},
\end{equation}
where $|\phi^+\rangle_{S_1S_2} = \frac{1}{\sqrt{2}} (|z+\rangle_{S_1}|z+\rangle_{S_2} + |z-\rangle_{S_1}|z-\rangle_{S_2})$ and $|\psi^-\rangle_{S_1S_2} = \frac{1}{\sqrt{2}} (|z+\rangle_{S_1}|z-\rangle_{S_2} - |z-\rangle_{S_1}|z+\rangle_{S_2})$ and the first spin is in possession of Charlie and the second of Debbie. The state can be obtained by applying rotation $(\mathbbm{1} \otimes e^{-\frac{i}{2} \theta \sigma_y})|\psi^-\rangle_{S_1S_2}$ to the singlet state $|\psi^-\rangle_{S_1S_2} = \frac{1}{\sqrt{2}} (|z+\rangle_{S_1}|z-\rangle_{S_2} - |z-\rangle_{S_1}|z+\rangle_{S_2})$, where $\theta$ is the angle of rotation of Debbie's spin around $y$-axis and $\sigma_y$ is a Pauli matrix. This particular choice of the state enables all measured observabes to be either of the Wigner's friend' type, or of Wigner's type.
 
For Alice and Bob, the overall state of the spins together with Charlie's and Debbie's laboratories is initially
\begin{equation}
|\Psi_0\rangle= |\psi\rangle_{S_1S_2} |0\rangle_C |0\rangle_D,
\end{equation}
in agreement with assumption 1. The state $|0\rangle_C |0\rangle_D$ of the two observers does not require further characterization, except for the description of observers capable of completing a measurement. 

Now Charlie and Debbie each perform a measurement of the respective spin along the $z$ direction. This measurement procedure is described as an unitary transformation from the point of view of Alice and Bob. We assume that after Charlie and Debbie complete their measurement, the overall state becomes
\begin{equation}
\label{ent}
|\tilde{\Psi}\rangle = -\sin\frac{\theta}{2} |\Phi^+\rangle + \cos\frac{\theta}{2} |\Psi^-\rangle,
\end{equation}
where 
\begin{align}
|\Phi^+\rangle = \frac{1}{\sqrt{2}} (|A_{up}|B_{up}\rangle + |A_{down}\rangle|B_{down}\rangle), \\ 
|\Psi^-\rangle = \frac{1}{\sqrt{2}} (|A_{up}\rangle |B_{down}\rangle - |A_{down}\rangle |B_{up}\rangle )
\end{align}
and 
\begin{align}
|A_{up}\rangle=|z+\rangle_{S_1}|C_{z+}\rangle_{C}, && |B_{up}\rangle=|z+\rangle_{S_2}|D_{z+}\rangle_{D}, \\
|A_{down}\rangle=|z-\rangle_{S_1}|C_{z-}\rangle_{C}, && |B_{down}\rangle =|z-\rangle_{S_2}|D_{z-}\rangle_{D}.
\end{align}

We take now $\theta=\pi/4$ and define two sets of (binary) observables which play the same role of a spin (Pauli) operators along the $z$ and $x$ axis respectively: $A_z=|A_{up}\rangle\langle A_{up}|-|A_{down}\rangle\langle A_{down}|$ and $A_x=|A_{up}\rangle\langle A_{down}|+|A_{down}\rangle\langle A_{up}|$ for Alice and similarly $B_z$ and $B_x$ for Bob. In the Bell experiment, Alice chooses between $A_1=A_z$ and $A_2=A_x$, whereas Bob chooses between $B_1=B_z $ and $B_2=B_x$. Note that Alice and Bob each choose between the friend's ($A_1$ and $B_1$) and Wigner's ($A_2$ and $B_2$) type of measurement. The Bell test with these measurement settings and state~(\ref{ent}) results in $S_Q=2\sqrt{2}$. The violation of the inequality implies that the conjunction of the assumptions (1-4) used to derive it is untenable.
\end{proof}

In Appendix we present a Greenberger-Horne-Zeilinger type of the theorem with three Wigners and three friends. There the discrepancy between quantum theory and the theories respecting (2-4) is no more of probabilistic, but deterministic nature. 

We conclude that Wigner, even as he has a clear evidence for the occurrence of a definite outcome in the friend's laboratory, {\em cannot} assume any specific value for the outcome to {\em coexist} together with the directly observed value of his outcome, given that all other assumptions are respected. Moreover, there is no theoretical framework where one can assign jointly the truth values to observational propositions of different observers (they cannot build a single Boolean algebra) under these assumptions. A possible consequence of the result  is that there cannot be \emph{facts of the world per se, but only relative to an observer} in agreement with Rovelli's relative-state interpretation~\cite{rovelli}, Quantum Bayesianism\footnote{Already in 1996 in the "Replies to Referee 4" of  Ref.~\cite{fuchs}, Fuchs drew a distinction between "facts for the agent" and "facts for everybody".}~\cite{chris}, as well as the (neo)-Copenhagen interpretation~\cite{brukner}. It is interesting to note that a similar view was expressed by Jammer as early as in 1974~\cite{jammer}, when he wrote that "the description of the state of a system, rather than being restricted to the particle (or systems of particles) under observation, expresses a relation between the particle and all the measurement devices involved." Other possible interpretations of violation of Bell's inequalities include violations of assumption 1 in collapse models~\cite{grw,diosi,penrose}, of assumption 2 in non-local hidden variable models such as de Broglie-Bohm theory~\cite{bohm} or of assumption 3. in superdeterministic theories~\cite{thooft}. The proper account of the result in the many-worlds interpretation should be found in the interpretation’s account of Bell's inequality violation~\cite{harvey,mateus} and points again to observer-dependent facts as they depend on the branch of the many worlds.

\subsection{Relation to the paper by Frauchiger and Renner, arXiv: 1604.07422}

Building upon works by Deutsch~\cite{deutsch}, Hardy~\cite{hardy92,hardy93} and Ref.~\cite{brukner} reviewed above, Frauchiger and Rennen~\cite{renner} proposed an "extended Wigner-friend thought experiment", from which they concluded that "single-world interpretations of quantum theory cannot be self-consistent". The implications of their argument have been discussed since then~\cite{veronika,sudbery,chris,bub}.

The claim of Ref.~\cite{renner} is based on an incompatibility proof stating that there cannot exist a physical theory T that would fulfill the following three properties (informal versions, see Ref.~\cite{renner} for details): 
\begin{description}
\item[(QT)] "{\it Compliance with quantum theory}: T forbids all measurement results that are forbidden by standard quantum theory (and this condition holds even if the measured
system is large enough to contain itself an experimenter)."

\item[(SW)] "{\it Single-world}:  T rules out the occurrence of more than one single outcome if an experimenter measures a system once."

\item[(SC)] "{\it Self-consistency}: T's statements about measurement outcomes are logically consistent (even if they are obtained by considering the perspectives of different experimenters)."
\end{description}
Property (QT) is essentially a weaker version of our assumption 1 where it is sufficient to require the validity of quantum theory for results with vanishing probability (as the argument is possibilistic, not probabilistic).  An example of a theory violating property (SW) is the many-worlds interpretation of quantum theory.

The argument combines a set of statements that involves different observers F$_1$, F$_2$, A and W and can be drawn on the basis of theory T: 
\begin{description}[labelindent=25pt,style=multiline,leftmargin=1.6cm]
\item [S$_1$] If F$_1$ sees $r=t$, then W sees $w \neq ok$.
\item [S$_2$] If F$_2$ sees $z=+$, then F$_1$ sees $r= t$.
\item [S$_3$] If A sees $x=ok$, then F$_2$ sees $z=+$.
\item [S$_4$] W sees $w= ok$ and is told by A that $x =ok$. 
\end{description}
The specific type of quantum state, measurements and outcomes involved in the argument is not relevant for further discussion and will be omitted here. 

Property (SC) is crucial in a step of the proof, where one combines "nested" statements (S$_1$ to S$_4$)~\cite{renato}. In the first step, the self-consistency property (SC) implies the following implication 
\begin{equation}
\mbox{S}_a \cap \mbox{S}_b \implies \mbox{S}_c
\label{reason1}
\end{equation}
where $\cap$ denotes logical "and", and the statements are of the type
\begin{description}[labelindent=25pt,style=multiline,leftmargin=1.6cm]
\item [S$_a$] Observer W assigns the truth value "true" to the statement: "A sees $x =ok$";
\item [S$_b$] Observer A assigns the truth value "true" to the statement: "If $x=ok$, then F$_2$ sees $z=+$";
\item [S$_c$] Observer W assigns the truth value "true" to the statement: "A concludes that F$_2$ sees $z=+$".
\end{description}
By repeating reasoning~(\ref{reason1}) in an iterative way, starting from statement S$_4$ to S$_1$, one arrives at a new statement
\begin{description}[labelindent=25pt,style=multiline,leftmargin=1.6cm]
\item [T] Observer W concludes that A concludes that F$_2$ concludes that F$_1$ concludes that $w \neq ok$.
\end{description}
It is important to note that this statement refers to W's conclusion about what other observers conclude when they apply T conditional on the outcomes they observe. It is not a statement  about his directly observed outcome. 

In the second step, the self-consistency property (SC) is used to arrive at an implication of the following type
\begin{equation}
\mbox{T} \implies \mbox{S}.
\label{reason2}
\end{equation}
where the implied statement is 
\begin{description}[labelindent=25pt,style=multiline,leftmargin=1.6cm]
\item [S] Observer W concludes that $w \neq ok$,
\end{description}
which stands in logical contradiction with W's directly observed outcome $w = ok$.

The second step is non-trivial. It enables to promote others' knowledge based on their observations to ones' own knowledge and then to put this "promoted knowledge" in logical comparison with ones' own knowledge gained through direct observation. Through implication (\ref{reason2}), the self-consistency property (SC) enables observational statements of other observers (A, F$_2$ and F$_1$) to be logically compared with ones (W) own. This has the same predictive power as a theoretical framework in which the truth values of statements of different observers can jointly be assigned and compared. To see this, denote statements S$_i$, $i=1,2,3$ as implications S$_1$: (P $\implies$ Q), S$_2$: (Q $\implies$ R) and
S$_3$: (R $\implies$ S), where P: "A sees $x=ok$", Q: "F$_2$ sees $z=+$", R: "F$_1$ sees $r=t$" and S: "W sees $w \neq ok$". Then, "collapsing" others' knowledge into W's knowledge via Eq.~(\ref{reason2}) is equivalent in its implications to considering all the statements as belonging to a single Boolean algebra (i.e. they are now all propositions of observer W, who can apply logical operations on them) for which one can use the transitivity of implication to arrive at [P $\cap$ (P $\implies$ Q) $\cap$ (Q $\implies$ R) $\cap$ (R $\implies$ S)] $\implies$ S. Statement S is again in a logical contradiction to W's directly observed outcome $w = ok$.

We have seen that an existence of a single Boolean algebra for truth values for observational statements of different observers is incompatible with the assumptions of "locality", "freedom of choices" and the predictions of quantum theory, which does not impose any constraints on the objects it is applied to. This might be interpreted as an indication that the strong conclusions implied by the theorem of Ref.~\cite{renner} rely on a too restrictive requirement of property (SC) on a physical theory. The requirement needs not only be fulfilled in quantum theory, but in other physical theories as well. An example was provided by Sudbery~\cite{sudbery}: In the special theory of relativity, due to time dilation, every inertial observer can claim that her/his clock ticks slower than that of a moving partner. This apparent contradiction in predictions of different observers is resolved when one realizes that the statements only have meaning with respect to the specific, observer-dependent measurement procedures that define "simultaneity". Similarly, the states referring to outcomes of different observers in a Wigner-friend type of experiment cannot be defined without referring to the specific experimental arrangements of the observers, in agreement with Bohr's idea of contextuality as formulated by him in 1963~\cite{bohr63}: "the unambiguous account of proper quantum phenomena must, in principle, include a description of all relevant features of experimental arrangement."

I conclude with a remark that the theorem by Frauchiger and Renner has deep conceptual value, as it points to the necessity to differentiate between ones' knowledge about direct observations and ones' knowledge about others' knowledge that is compatible with physical theories. It is likely that understanding this difference will be an important ingredient in further development of the method of Bayesian inference in situations as in the Wigner-friend experiment. 

I acknowledge helpful discussions with Mateus Ara{\'u}jo, Veronika Baumann, Ad{\'a}n Cabello, Giulio Chiribella, Christopher Fuchs, Borivoje Daki{\'c}, Philipp H{\"o}hn, Nikola Paunkovi{\'c}, L{\'i}dia del Rio, R{\"u}diger Schack and Stefan Wolf. I would like to especially acknowledge the fruitful discussions with Renato Renner and thank him for providing notes summarizing that discussion. 

\section{Appendix}
\label{Supply}

The Bell theorem from the main text can be extended to a Greenberger-Horne-Zeilinger (GHZ) version~\cite{ghz} with three friends and three Wigners. Since the incompatibility of assumptions 1.-4. is not of a probabilistic but rather of a deterministic nature, this version of the theorem completely bypasses  any use of the notion of probability, similarly to the version by Frauchiger and Renner~\cite{renner}. 

Consider three spatially separated observers (Wigners), Alice, Bob and Cleve. They each perform a measurement on a subsystem of a tripartite system. Each of the subsystems includes a further observer, Debbie, Eric, and Fiona (Wigner's friends), who perform a Stern-Gerlach measurement of spin along $x$ of their respective spin-1/2 particles. Alice measures Debbie and her spin particle, Bob measures Eric and his spin particle, and finally, Cleve measures Fiona and her spin particle. We consider a GHZ test where Alice chooses between two measurement settings: $A_1$ and $A_2$, Bob between $B_1$ and $B_2$ and Cleve between $C_1$ and $C_2$. The assumptions 2., 3. and 4. imply that $A_1$, $A_2$, $B_1$, $B_2$, $C_1$ and $C_2$ have predefined values $+1$ or $-1$.

Define  $\hat{A}_x=|A_{up}\rangle\langle A_{up}|-|A_{down}\rangle\langle A_{down}|$ and $\hat{A}_y=i(|A_{up}\rangle\langle A_{down}|-|A_{down}\rangle\langle A_{up}|)$ for Alice and similarly $\hat{B}_x$ and $\hat{B}_y$ for Bob and $\hat{C}_x$ and $\hat{C}_y$ for Cleve, where 
\begin{align}
|A_{up}\rangle=|x+\rangle_{A1}|D_{x+}\rangle_{A2}, && |B_{up}\rangle=|x+\rangle_{B1}|E_{x+}\rangle_{B2} &&  |C_{up}\rangle=|x+\rangle_{C1}|F_{x+}\rangle_{C2}, \\
|A_{down}\rangle=|x-\rangle_{A1}|D_{x-}\rangle_{A2}, && |B_{down}\rangle =|x-\rangle_{B1}|E_{x-}\rangle_{B2} &&  |C_{down}\rangle=|x-\rangle_{C1}|F_{x-}\rangle_{C2}.
\end{align}
In the GHZ test, we choose $\hat{A}_1=\hat{A}_x$, $\hat{A}_2=\hat{A}_y$ for Alice and similarly for Bob and Cleve. 

Assume that Alice, Bob and Cleve perform these measurements on a shared GHZ state:
\begin{equation}
\label{GHZ}
|\Psi_{GHZ}\rangle_{ABC}= \frac{1}{\sqrt2} \left(|A+\rangle|B+\rangle |C+\rangle - |A-\rangle |B-\rangle |C-\right),
\end{equation}
where due to assumption 1 we presume that such a state can be prepared and $|A \pm \rangle = \frac{1}{\sqrt{2}}(|A_{up}\rangle  \pm |A_{down}\rangle)$, $|B \pm \rangle = \frac{1}{\sqrt{2}}(|B_{up}\rangle  \pm |B_{down}\rangle)$ and $|C\pm \rangle = \frac{1}{\sqrt{2}} (|C_{up}\rangle  \pm |C_{down}\rangle)$.

In order to reproduce perfect correlations in the GHZ state, the predefined values need to satisfy $A_x B_y C_y=A_y B_x C_y=A_y B_y C_x=1$. These equations imply then that $A_x B_x C_x =1$, however one finds the opposite result in quantum mechanics: $\hat{A}_x \hat{B}_x \hat{C}_x|\Psi_{GHZ}\rangle_{ABC}=-|\Psi_{GHZ}\rangle_{ABC}$.

\end{document}